\documentclass[]{lipics-v2021}
\usepackage[utf8]{inputenc}
\usepackage{color} 

\hideLIPIcs \nolinenumbers %uncomment to remove references to LIPIcs series 

\usepackage{hyperref}
\usepackage{graphicx}

\newcommand{\journal}[1]{#1} % Show journal
\newcommand{\conf}[1]{}      % Hide conference

\newcommand{\RE}{\mathbb{R}}            % real space
\newcommand{\TT}{\mathcal{T}}
\newcommand{\IT}{\mathbf{T}}
\newcommand{\II}{\mathcal{I}}
\newcommand{\UP}{\blacktriangledown}
\newcommand{\DW}{\blacktriangle}
\newcommand{\extraUP}{\textcolor{black!35}{\blacktriangledown}}
\newcommand{\extraDW}{\textcolor{black!35}{\blacktriangle}}

\title{Shadoks Approach to Parallel Reconfiguration of Triangulations}

\supplementdetails[subcategory={Source Code}]{}{https://github.com/gfonsecabr/shadoks-CGSHOP2026}
\relatedversion{\href{https://arxiv.org/abs/}{arxiv.org/abs/}}

 	\author{Guilherme D. da Fonseca}
 	{LIS, Aix-Marseille Université}
 	{guilherme.fonseca@lis-lab.fr}
 	{https://orcid.org/0000-0002-9807-028X}
 	{}
 	
 	\author{Fabien Feschet}
 	{LIMOS, Université Clermont Auvergne}
 	{fabien.feschet@uca.fr}
 	{https://orcid.org/0000-0001-5178-0842}
 	{}

 	\author{Yan Gerard}
 	{LIMOS, Université Clermont Auvergne}
 	{yan.gerard@uca.fr}
 	{https://orcid.org/0000-0002-2664-0650}
 	{}

 	\authorrunning{G. D. da Fonseca, F. Feschet, and Y. Gerard}

\EventEditors{Erin W. Chambers and Joachim Gudmundsson}
\EventNoEds{2}
\EventLongTitle{39th International Symposium on Computational Geometry (SoCG 2023)}
\EventShortTitle{SoCG 2023}
\EventAcronym{SoCG}
\EventYear{2023}
\EventDate{June 12--15, 2023}
\EventLocation{Dallas, Texas, USA}
\EventLogo{socg-logo.pdf}
\SeriesVolume{258}
\ArticleNo{67}
\category{CG Challenge}

\Copyright{Guilherme D. da Fonseca, Fabien Feschet, and Yan Gerard}

\acknowledgements{We would like to thank the Challenge organizers and other competitors for their time, feedback, and making this whole event possible. We would like to thank Hélène Toussaint, Raphaël Amato, Boris Lonjon, and William Guyot-Lénat from LIMOS, as well as the Qarma and TALEP teams and Manuel Bertrand from LIS, who continue to make the computational resources of the LIMOS and LIS clusters available to our research. We would also like to thank Aldo Gonzalez-Lorenzo for the very useful discussion on SAT models.}

\begin{document}

\maketitle

\begin{abstract}
We describe the methods used by Team Shadoks to win the CG:SHOP 2026 Challenge on parallel reconfiguration of planar triangulations.
\journal{An instance is a collection of triangulations of a common point set. We must select a center triangulation and find short parallel-flip paths from each input triangulation to the center, minimizing the sum of path lengths.}
Our approach combines exact methods based on SAT with several greedy heuristics, and also makes use of SAT and MaxSAT for solution improvement.
\journal{We present a SAT encoding for bounded-length paths and a global formulation for fixed path-length vectors. We discuss how these components interact in practice and summarize the performance of our solvers on the benchmark instances.}
\end{abstract}

%--------------------------------------------------------------------
\section{Introduction} \label{s:intro}
%--------------------------------------------------------------------

\begin{figure}[b]
 \centering
 \includegraphics[angle=90,height=4cm]{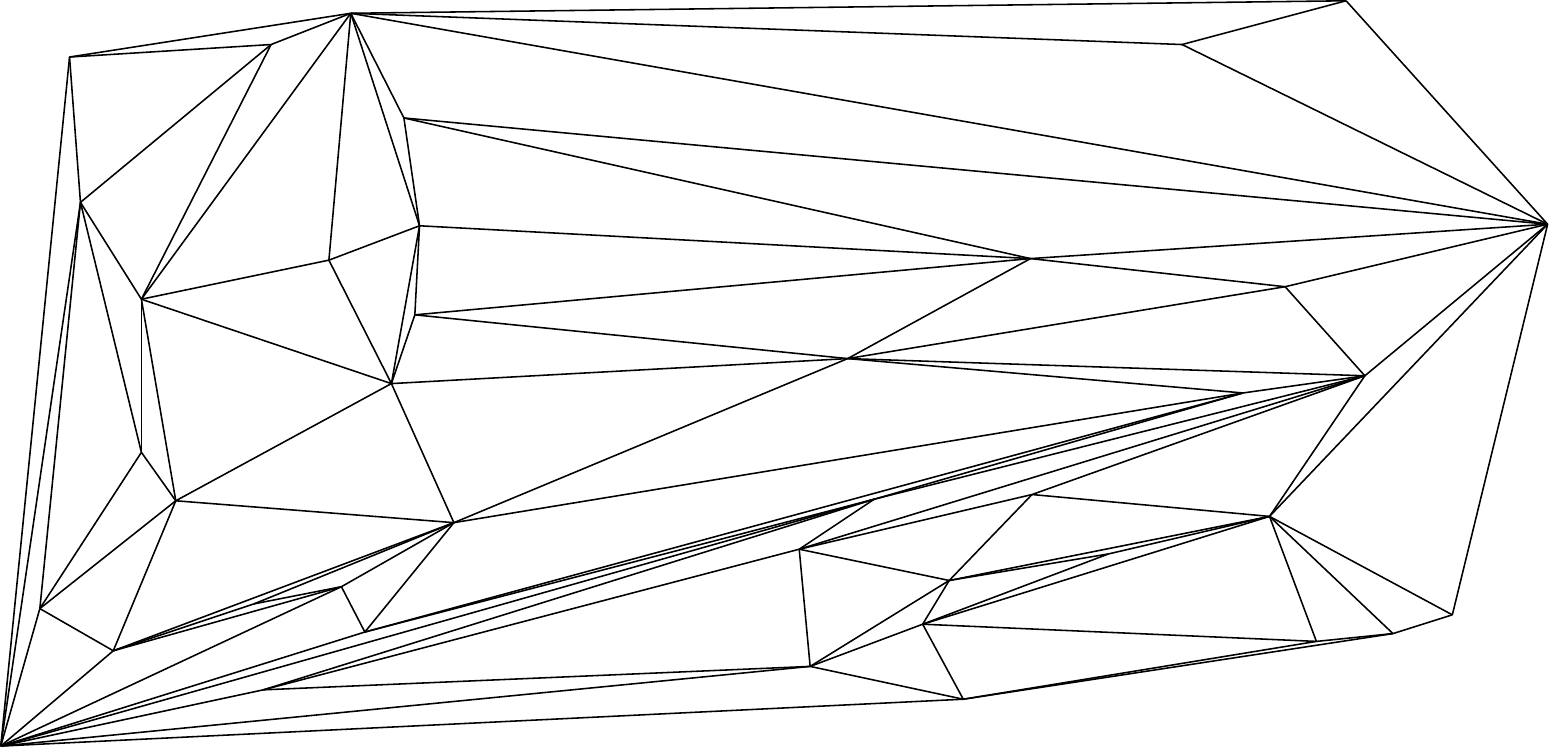}\hfill
 \includegraphics[height=4cm]{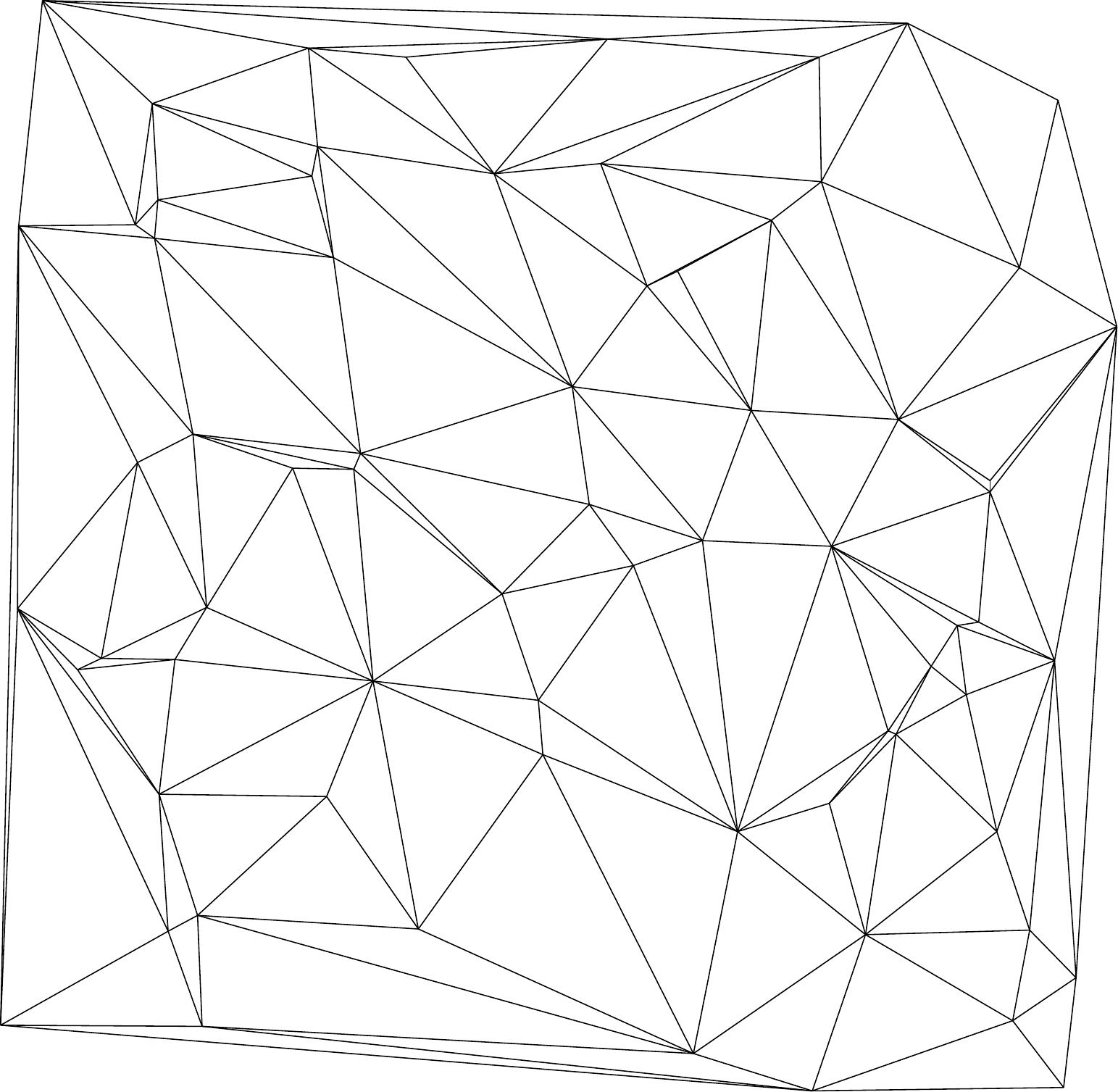}\hfill
 \includegraphics[angle=90,height=4cm]{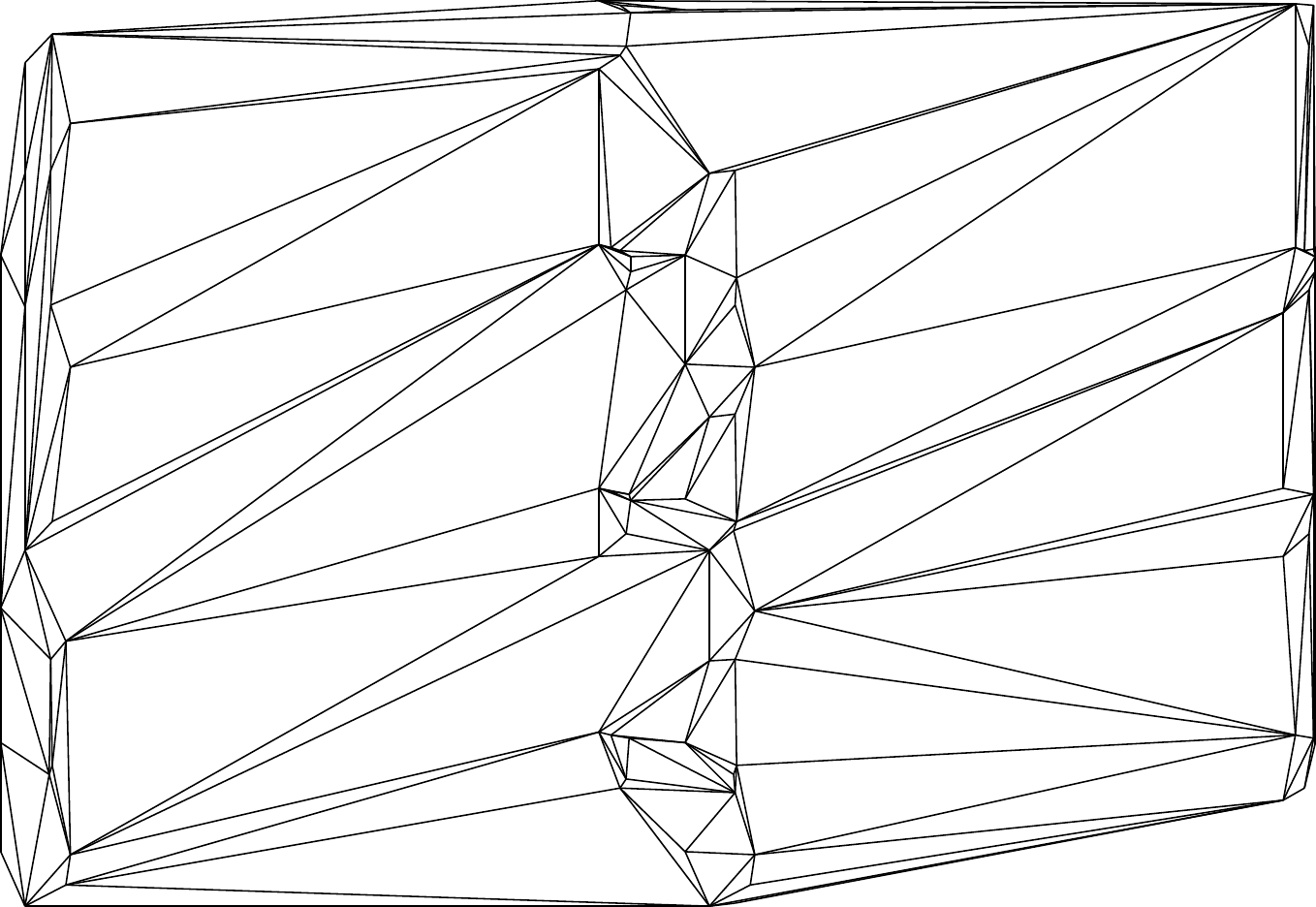}\hfill
 \includegraphics[height=4cm]{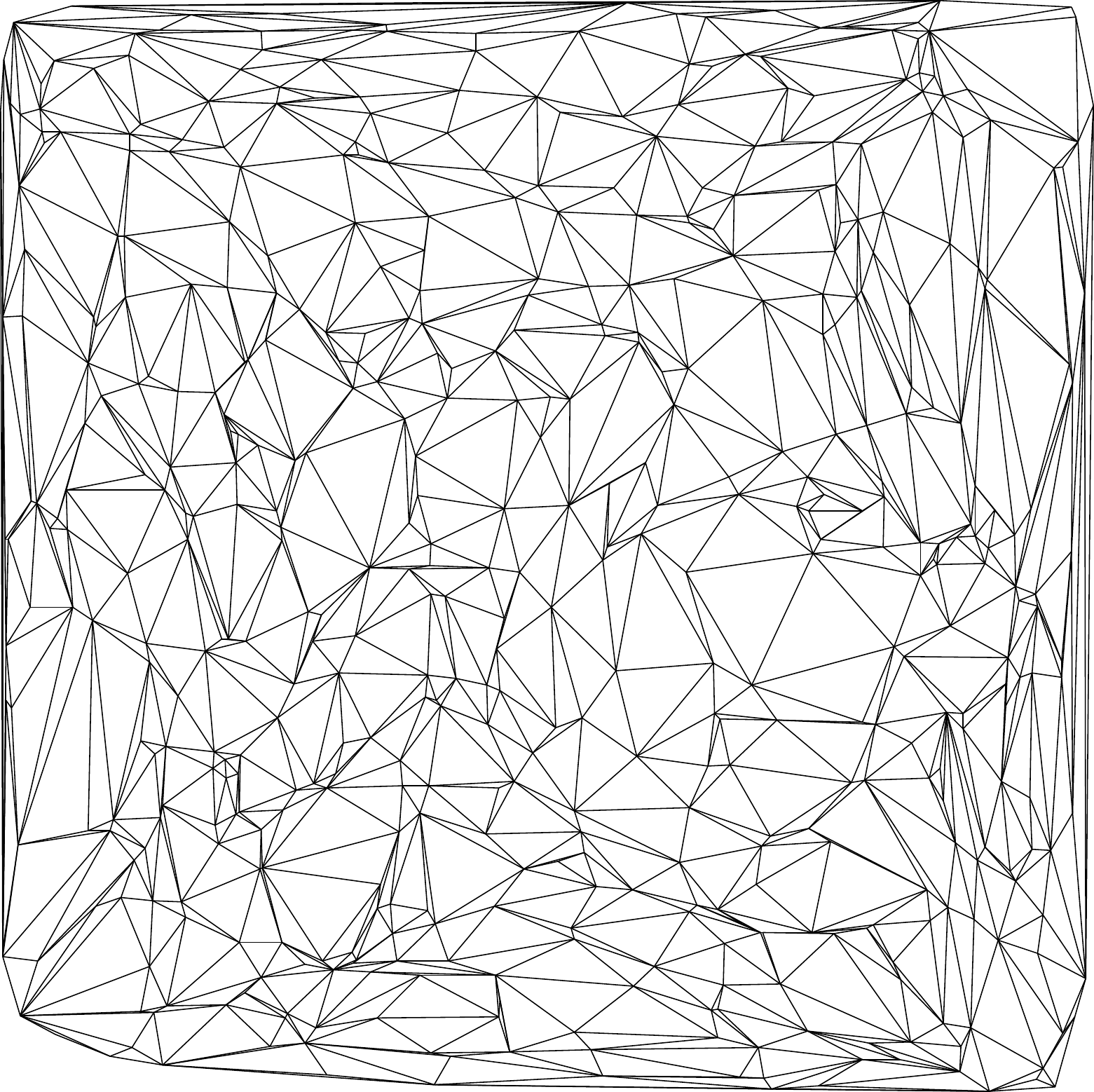}
 \caption{Our best centers to instances \texttt{random\_78\_40\_10}, \texttt{woc-70-random-9a7d18d3}, \texttt{woc-90-tsplib}, and \texttt{rirs-500-50-23d00ec5}, respectively.}
 \label{f:centers}
\end{figure}

The CG:SHOP Challenge is an annual competition in geometric optimization. In 2026, the challenge focuses on a reconfiguration problem between planar triangulations. Our team, called \emph{Shadoks}, won first place with the best solution (among the 28 participating teams) to $249$ instances out of $250$ instances and provably optimal solutions to $189$ instances.

In this paper, we outline the exact methods \journal{and heuristics} that we employed. \conf{The heuristic solvers are only presented in the full version due to space limitations.} We start with some definitions that allow us to describe the problem. Throughout, we consider triangulations of a common point set $S \subset \RE^2$.

Given a triangulation $T$, a \emph{unit flip} is the operation that removes an edge $e \in T$ and adds an edge $e'$.
The unit flip of the edge $e=uv$ considers the empty convex quadrilateral $u,u_2,v,v_2$ and replaces its diagonal $uv$ it by the other diagonal $u_2v_2$ (Figure~\ref{f:flipvar}).
% Notice that the edge $e$ must cross $e'$ which means that the quadrilateral $u,u_2,v,v_2$  must be convex. 

Similarly, a \emph{parallel flip} removes a set of edges $E \subset T$ and adds a set of edges $E'$, in a way that $T' = T \setminus E \cup E'$ is a triangulation, with the condition that no two edges of $E$ are in the same triangle in $T$.
A \emph{path} of \emph{length} $\ell$ is a sequence of triangulations $T_0,\ldots,T_\ell$ such that for all $i$, the triangulation $T_{i+1}$ is obtained from $T_i$ by performing a parallel flip.

An \emph{instance} is a set $S \subset \RE^2$ of $n$ points and a set $\TT = \IT_1,\ldots,\IT_{|\TT|}$ of triangulations of $S$, called \emph{input triangulations}. A \emph{solution} is a set of paths $P_1,\ldots,P_{|\TT|}$ such that $P_i$ starts at $\IT_i$ for all $i$ and all paths end in a common triangulation called \emph{center}.
The goal is to find a solution that minimizes the \emph{objective value} defined as the sum of the lengths of its paths.

During the competition, the organizers provided a total of $250$ instances, with $n$ ranging from $15$ to $12{,}500$ points and $|\TT|$ ranging from $2$ to $200$ triangulations. The $250$ instances are divided into three classes: $100$ \texttt{random} instances, $101$ \texttt{woc} instances, and $49$ \texttt{rirs} instances. The former two instances have up to $320$ points and $2$ to $20$ input triangulations (hence, we call them \texttt{small} instances), while the latter have $500$ to $12500$ points and $20$ to $200$ input triangulations.
The centers of some of our best solutions are presented in Figure~\ref{f:centers}. Additional details about the challenge can be found in the organizers' survey paper~\cite{survey}. 

Our best solvers heavily rely on the SAT solver \texttt{CaDiCal}~\cite{cadical} and the MaxSAT solver \texttt{EvalMaxSAT}~\cite{evalmaxsat}. 
\journal{Nevertheless, we also developed heuristics that do not rely on any external solver, which are important to find initial solutions to large instances, which are then improved by roughly $10\%$ using SAT and MaxSAT solvers.}
\conf{Nevertheless, we also developed heuristics that do not rely on any external solver, which are important to find initial solutions to large instances, which are then improved by roughly $10\%$ using SAT and MaxSAT solvers., which are presented in the full version.} 

Other team strategies include SAT solvers for small instances~\cite{battini2026eth}, heuristics~\cite{lee2026cghunters}, and local search and simulated annealing~\cite{conradi2026engineering} to improve solutions.

\journal{
We describe our exact algorithms in Section~\ref{s:exact}, the heuristics in Section~\ref{s:heuristics}, and discuss the results we obtained in Section~\ref{s:results}. Concluding remarks and open problems are presented in Section~\ref{s:conclusion}.
}

%--------------------------------------------------------------------
\section{Exact Algorithms} \label{s:exact}
%--------------------------------------------------------------------

\journal{
This section describes all elements of our exact solver, many of which are also used in the heuristic solvers. We first show how to use a SAT solver to compute shortest paths between two triangulations (Section~\ref{s:satDistance}). We then show how to extend this result to test if a solution with a list of path lengths exists (Section~\ref{s:satSolution}). We show how to obtain lower bounds in Section~\ref{s:lowerbound} and put the previous elements together to describe our exact solver in Section~\ref{s:exactSolver}.
}

%....................................................................
\subsection{Path SAT Formulation} \label{s:satDistance}

Next, we describe a SAT formulation for the following decision problem. The input is a set $S$ of $n$ points, an integer $\ell$ and two triangulations $T_0,T_\ell$. The output is whether there exists a path $T_0,\ldots,T_\ell$ of length $\ell$.

We define two types of variables. For $i = 0,\ldots,\ell$ and for $u \neq v \in S$, we define an \emph{edge variable} $e(u,v,i)$. The variable $e(u,v,i)$ represents that the edge $uv$ is in the triangulation $T_i$. There are $O(n^2\ell)$ edge variables. It would be possible to define a SAT formulation using only such variables. However, a SAT formulation that performed much better in our experiments uses a second type of variable.

\begin{figure}[ht]
 \centering
 \includegraphics[scale=.85]{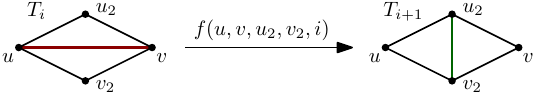}
 \caption{Illustration of a flip and the associated variable $f(u,v,u_2,v_2,i)$.}
 \label{f:flipvar}
\end{figure}

We say that a convex quadrilateral is \emph{empty} if it contains no point of $S$ except for its vertices. For $i = 0,\ldots,\ell$-1 and for an empty convex quadrilateral $u,u_2,v,v_2$, we introduce \emph{flip variables} $f(u,v,u_2,v_2,i)$, which is true if and only if $uv$ is in triangulation $T_i$ and $u_2v_2$ is in triangulation $T_{i+1}$, as shown in Figure~\ref{f:flipvar}.
Notice that if the points are uniformly distributed, then the number of empty convex quadrilaterals is $\Theta(n^2)$~\cite{quadrilateral}, which means that for uniformly distributed points, the number of flip variables is also $O(n^2\ell)$. However, the number of flip variables is $\Theta(n^4\ell)$ if the points are in convex position (which is not the case for the challenge instances).
Next, we describe the different types of clauses.

\subparagraph{Start and target.} For every edge variable $e(u,v,0)$, we have the clause $e(u,v,0)$ if $uv \in T_0$ and $\neg  e(u,v,0)$ if $uv \notin T_0$.
For every edge variable $e(u,v,\ell)$, we have the clause $e(u,v,\ell)$ if $uv \in T_\ell$ and $\neg  e(u,v,\ell)$ if $uv \notin T_\ell$.

\subparagraph{Flips need edges.} For every flip variable $f(u,v,u_2,v_2,i)$, we have the  $5$ binary CNF clauses translating the implication \[f(u,v,u_2,v_2,i) \implies e(u,v,i) \land e(u,v_2,i) \land e(u,u_2,i) \land e(v,v_2,i) \land e(v,u_2,i).\]

\subparagraph{Flips keep edges.} Similarly, or every flip variable $f(u,v,u_2,v_2,i)$, we have the $5$ binary CNF clauses translating the implication \[f(u,v,u_2,v_2,i) \implies e(u_2,v_2,i+1) \land e(u,v_2,i+1) \land e(u,u_2,i+1) \land e(v,v_2,i+1) \land e(v,u_2,i+1).\] 

\subparagraph{Flips flip edges.} For every flip variable $f(u,v,u_2,v_2,i)$, we have the binary clauses translating \[f(u,v,u_2,v_2,i) \implies \neg e(u,v,i+1).\]

\subparagraph{Edge changes require flips.} The last type of clause is the only one that has more than $2$ variables in CNF form. It states that if the edge variable changes from triangulation $i$ to $i+1$, then there must be a flip. The $\bigvee$ below considers all values of the subscript that define existing flip variables. We have two such clauses for each edge variable:
\[e(u,v,i) \land \neg e(u,v,i+1) \implies \bigvee_{u_2,v_2}f(u,v,u_2,v_2,i) \text { and }\]
\[\neg e(u_2,v_2,i) \land e(u_2,v_2,i+1) \implies \bigvee_{u,v}f(u,v,u_2,v_2,i).\]

\subparagraph{Eliminating variables and clauses.}
The number of variables and clauses grows very fast, even though the number of clauses is linear in the number of variables. Next, we show how to eliminate many variables from the model. All eliminated variables are defined as \emph{false} and the clauses that become tautologies are eliminated. If a CNF clause becomes empty, then the problem is unsatisfiable.

\journal{Notice that we can eliminate $e(u,v,0)$ for $uv \notin T_0$ and $e(u,v,\ell)$ for $uv \notin T_\ell$. This is, however, only a special case of a more general rule.}
The following theorem is easy to prove and implies that $\Omega(\log n)$ parallel flips are sometimes necessary to reconfigure two triangulations of $n$ points, even when the points are in convex position. We say that two segments \emph{cross} if they intersect at a point that is not an endpoint of either segment.

\begin{theorem} \label{t:cross}
Consider two triangulations $T,T'$ of $S$ such that a parallel flip transforms $T$ into $T'$ and a segment $s$ with endpoints in $S$. Let $\chi,\chi'$ respectively denote the number of edges of $T,T'$ crossed by $s$. We then have $\chi' \geq \lfloor \chi / 2 \rfloor$.
\end{theorem}
\journal{
\begin{proof}
Consider the sequence $L$ of edges of $T$ crossed by $s$ in the order they cross the segment $s$. A parallel flip cannot remove two consecutive edges of $L$ because they share a triangle, hence the theorem follows.
\end{proof}
}

Consequently, we only define the variable $e(u,v,i)$ when $uv$ crosses strictly less than $2^i$ edges of $T_0$ and strictly less than $2^{\ell-i}$ target edges.
We only define flip variables when a certain set of edge variables is defined:  $f(u,v,u_2,v_2,i)$ is only defined when $uv, uv_2, uu_2,vv_2, vu_2$, are all defined at $i$ and $u_2v_2, uv_2, uu_2,vv_2, vu_2$, are all defined at $i+1$.

%....................................................................
\subsection{Solution SAT Formulation} \label{s:satSolution}

Next, we describe a SAT formulation for the following decision problem. Recall that an instance is a set $S$ of $n$ points and a list $\TT$ of input triangulations $\IT_1,\ldots,\IT_{|\TT|}$. The input of the decision problem is an instance and $|\TT|$ integers $\ell_1,\ldots,\ell_{|\TT|}$. The output is whether there exists a solution $P_1,\ldots,P_{|\TT|}$ such that path $P_i$ has length $\ell_i$ for all $i$.

We model the $|\TT|$ paths $P_1,\ldots,P_{|\TT|}$ independently as before, starting path $P_i$ at the input triangulation $\IT_i$. The final triangulation of each path is unknown, but the same edge variables are used for the final triangulation of every path, since a valid solution requires that all paths end in the same triangulation. It is easy to see that the SAT formulation is satisfiable if and only if there exists a solution with the given lengths.

\subsection{Lower Bound} \label{s:lowerbound}

In order to obtain an exact solution to an instance $\II$, we start by computing a lower bound to its objective value. We say that the \emph{distance} between two triangulations $T,T'$ is the length of the shortest path from $T$ to $T'$. We create a complete directed graph $G(\II)$ with edge lengths as follows. The vertices are the triangulations $\TT$ and the length of each edge is the distance between the corresponding triangulations. A \emph{cycle packing} of $G$ is a collection of vertex-disjoint directed cycles, {i.e.} a subset of edges such that each vertex has at most one outgoing and at most one incoming edge in the subset. The graph is directed to allow for cycles with only $2$ edges. The \emph{length of a cycle} is the sum of the lengths of its edges, and the \emph{length of a cycle packing} is the sum of the lengths of its cycles. The maximum length cycle packing can be solved in polynomial time using a reduction to maximum weight bipartite matching~\cite{cyclepacking}. \journal{We have the following theorem.} \conf{The following theorem is easy to show.}

\begin{theorem} \label{t:cyclepacking}
Given an instance $\II$, the objective value of a solution is at least the length of any cycle packing of $G(\II)$ divided by $2$.
\end{theorem}
\journal{
\begin{proof}
Let $d_{i,j}$ denote the distance between the input triangulations $\IT_i,\IT_j$. Let $C$ be a potential center and let $r_i$ be the distance from $C$ to $\IT_i$. Consider a cycle $\IT_1,\ldots,\IT_k$ in the cycle packing. By triangle inequality $d_{i,i+1} \leq r_i+r_{i+1}$ with indices taken modulo $k$ (see Figure~\ref{f:cyclepack} (b)). Summing over the inequalities for $i$ from $1$ to $k$, we have that the length of the cycle is at most $2 \sum_i r_i$. Applying the same argument to every cycle, the theorem follows.
\end{proof}
}

\begin{figure} [ht]
    \centering
    \includegraphics[scale=.85]{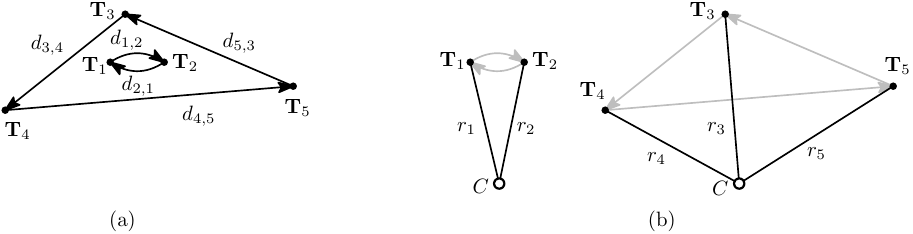}
    \caption{(a) A cycle packing. (b) Illustration of the proof. In this example,$d_{1,2} \leq r_1+r_2$, $d_{2,1} \leq r_2+r_1$, $d_{3,4} \leq r_3+r_4$, $d_{4,5} \leq r_4+r_5$, and $d_{5,3} \leq r_5+r_3$ by triangle inequality.}
    \label{f:cyclepack}
\end{figure}

%....................................................................
\subsection{The Exact Solver} \label{s:exactSolver}

First, we use the exact path formulation from Section~\ref{s:satDistance} to calculate the distance between all $\binom{\TT}{2}$ pairs of input triangulations using a SAT solver (in our case, \texttt{CaDiCal}~\cite{cadical}).
It is easy to formulate the problem of finding a maximum length cycle packing as a weighted MaxSAT problem, which provides a lower bound $b$ to the objective value (see Section~\ref{s:lowerbound}). We solve this problem using a weighted MaxSAT solver (in our case \texttt{EvalMaxSAT}~\cite{evalmaxsat}). We then use backtracking to list all integer solutions to $\ell_0,\ldots,\ell_{|\TT|} = b$ that satisfy $\ell_i+\ell_j \geq \textrm{distance}(T_i,T_j)$. We use the SAT formulation from Section~\ref{s:satSolution} to test the existence of a solution with the given lengths $\ell_0,\ldots,\ell_{|\TT|}$, again using the \texttt{CaDiCal} SAT solver. If a solution is found, then it is optimal. Otherwise, we increment $b$ and repeat. Notice that $b$ is always a lower bound to the objective value. Hence, if a solution obtained by a heuristic attains this lower bound, then it is optimal.

\journal{
%--------------------------------------------------------------------
\section{Heuristics} \label{s:heuristics}
%--------------------------------------------------------------------

In this section, we describe different approaches that can be used when we do not need to guarantee the optimality of the solution. In Section~\ref{s:happy}, we present a conjecture that allows us to significantly increase the performance of the SAT solver. In Section~\ref{s:crossLB}, we show how to further increase the performance of the SAT solver using a heuristic coupled with some instance-independent preprocessed data. In Section~\ref{s:improve}, we show how to use the SAT solver to improve existing solutions. In Sections~\ref{s:pathHeuristic} and~\ref{s:solutionHeuristic}, we respectively show how to compute short paths and good solutions without a SAT solver.
}

%....................................................................
\subsection{Happy Edges Conjecture} \label{s:happy}

The happy edges conjecture~\cite{happy} is a general conjecture that is \emph{false} for some reconfiguration problems and \emph{true} for others.

\begin{conjecture} \label{c:happy}
For any pair of configurations $T,T'$, there exists a shortest path between $T,T'$ where the edges that are common to both $T$ and $T'$ appear in all intermediate configurations.
\end{conjecture}

The conjecture is \emph{false} for triangulations under unit flips and arbitrary points~\cite{Pil14} but \emph{true} when the points are in convex position~\cite{STT86}. Our experiments lead us to believe that the conjecture is \emph{true} for parallel flips, as we could not find a counterexample.

\journal{Enforcing that an edge never disappears and later reappears along a path actually makes the SAT formulation harder to solve. However, there are some implications of the conjecture that are very useful to make the SAT formulation shorter and easier to solve.}

When computing a path of length $\ell$ from $T_0$ to $T_\ell$ using SAT, for every edge $uv$ that appears in both $T_0$ and $T_\ell$, we add clauses $e(u,v,i)$ that force the variable to be true for all $i$. More importantly, we then eliminate every edge variables corresponding to edges that cross $uv$. The same idea can be applied to the SAT formulation that finds a solution, but then only the edges that appear in all input triangulations are forced to be true for all $i$, and again the edges that cross them are eliminated.

Furthermore, when computing a path, for every edge $uv \in T_\ell$, we eliminate flip variables that remove $uv$, {i.e.} $f(u,v,u_2,v_2,i)$ for all $u_2,v_2,i$. Similarly, for every edge $u_2v_2 \in T_0$, we eliminate flip variables that insert $u_2v_2$, that is $f(u,v,u_2,v_2,i)$ for all $u,v,i$.

\journal{
%....................................................................
\subsection{Crossing Lower Bound} \label{s:crossLB}

Let $b(uv,T_0)$ denote the number of parallel flips needed to obtain an edge $uv$ starting at a triangulation $T_0$.
Clearly, when creating the SAT formulation for a path $T_0,T_1,\ldots$ we only need to define the edge variable $e(u,v,i)$ for $i \geq b(T_0,uv)$.
Theorem~\ref{t:cross} implies that if an edge $uv$ crosses $\chi(uv,T_0)$ segments of $T_0$, then $b(uv,T_0) \geq \lceil\log_2(\chi(uv,T_0)+1)\rceil$. This bound is tight when all the edges of $T_0$ that cross $uv$ share a common endpoint and the endpoints are in convex position (Figure~\ref{f:crosslb1}). However, there are different ways in which the edges of $T_0$ may cross $uv$ that may imply higher values of $b(uv,T_0)$.

\begin{figure} [ht]
    \centering
    \includegraphics[scale=.85,page=2]{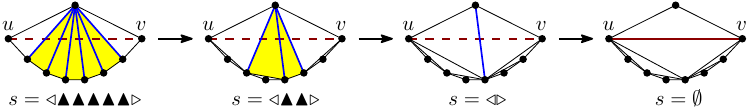}
    \caption{A path of length $3$ to insert an edge that had $6$ crossings.}
    \label{f:crosslb1}
\end{figure}

We consider estimations of $b(uv,T_0)$ based on the sequence of triangles that contain an upper or a lower edge (Figure~\ref{f:crosslb2}). An edge $uv$ that crosses $\chi(uv,T_0)$ segments of $T_0$ is translated into a \emph{string} $s=s(uv,T_0)$ of $\chi(uv,T_0) + 1$ symbols in the alphabet $\{\UP,\DW,\triangleleft,\triangleright\}$, according to which sides of $uv$ contain the edges that are not crossed by $uv$ in each triangle. 

\begin{figure} [ht]
    \centering
    \includegraphics[scale=.85,page=3]{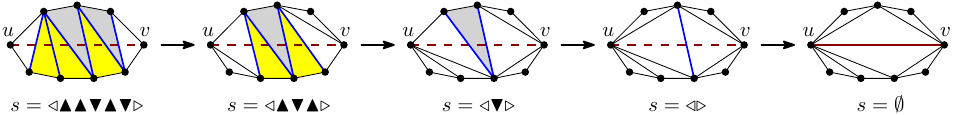}
    \caption{A path of length $4$ to insert an edge that had $6$ crossings. Each triangle is labeled and colored as containing an upper or a lower edge.}
    \label{f:crosslb2}
\end{figure}

The unit flips on $T_0$ have equivalent \emph{replacements} on $s$ that replace the substring on the \emph{left-hand side} by the substring on the \emph{right-hand side}. We define a \emph{substring} as a contiguous subsequence. Flipping the two extreme triangles translates to the \emph{extreme replacements}
$\triangleleft\DW \rightarrow \triangleleft$,
$\triangleleft\UP \rightarrow \triangleleft$,
$\DW\triangleright \rightarrow \triangleright$, and
$\UP\triangleright \rightarrow \triangleright$.
Flipping intermediate triangles with the same orientation translates to
$\DW \DW\rightarrow \DW$ and
$\UP\UP \rightarrow \UP$,
while flipping intermediate triangles of different orientations translates to
$\UP\DW \rightarrow \DW\UP$ and 
$\DW\UP \rightarrow \UP\DW$.
The last flip to insert the edge $uv$ is
$\triangleleft\triangleright \rightarrow \emptyset$.
There are other replacements that may increase the string length such as
$\DW \rightarrow \DW\DW$, which consist of exchanging the left-hand and right-hand side of some aforementioned replacements, but we show that we can always obtain shortest paths without using them.

A parallel flip means that we may apply several replacements simultaneously as long as their left-hand sides correspond to disjoint substrings. We call the application of several replacements a \emph{rewriting}, and use the notation $s \Rightarrow s'$ to say that $s'$ is a rewriting of $s$.
A \emph{rewriting sequence} of length $\ell$ is a sequence of $\ell+1$ strings connected by rewriting operations and ending at the empty string $\emptyset$. We want to find shortest rewriting sequences. 
Let $b(s)$ be the length of the shortest rewriting sequence of the string $s$. 

Notice that if the points are not in convex position, then some replacements may correspond to invalid flips. Hence, assuming convex position provides a lower bound $b(s(uv,T_0)) \leq b(uv,T_0)$. Furthermore, the replacements assume that there is an unbounded number of points, so that new triangles may be created freely.

We say that a string $s$ is a \emph{subword} of a string $s'$ if $s$ may be obtained from $s'$ by removing $\UP,\DW$ symbols without changing the order of the remaining symbols.
We prove the following results.

\begin{theorem}\label{temp0}
    If $s$ is a subword of $s'$, then $b(s) \leq b(s')$. 
\end{theorem}

\begin{proof}
Let $s$ be a subword of $s'$ and consider a shortest rewriting sequence $s'_0 \Rightarrow \cdots \Rightarrow s'_{b(s')}$ of the string $s'$ (we have $s'_0=s'$ and $s'_{b(s')}=\emptyset$). We show how to build a rewriting sequence $s_0\Rightarrow \cdots\Rightarrow s_{b(s')}$ of $s$ of the same length (we have $s_0=s$ and $s_{b(s')}=\emptyset$).  

We show that $s_i$ is a subword of $s'_i$ for all $i$. 
The proof proceeds by induction. The statement is true for $i=0$, since $s$ is a subword of $s'$.
Assume that $s_i$ is a subword of $s'_i$. We build $s_{i+1}$ as follows: 
the replacements applied to $s'_ {i}$ involving symbols which are all in $s_i$ are applied to $s_i$ to obtain $s_{i+1}$. All the other replacements are not taken into account to obtain $s_{i+1}$. 

We now show that the string $s_{i+1}$ computed this way is a subword of $s'_{i+1}$. There are different types of replacements that we applied to $s'_ i$ but not to $s_i$.
For convenience, the symbols of a word $s_i$ that are not in the subword $s'_i$ are called \emph{extra symbols} and shown as gray symbols $\extraUP,\extraDW$.

\begin{itemize}
    \item Replacements with one extra symbol
    of the form
    $\extraUP \UP \rightarrow \UP$, 
    $\UP\extraUP \rightarrow \UP$,
    $\extraDW\DW \rightarrow \DW$,  
    $\DW \extraDW \rightarrow \DW$,
    $\triangleleft\extraUP \rightarrow \triangleleft$,
    $\triangleleft\extraDW \rightarrow \triangleleft$,
    $\extraUP \triangleright \rightarrow \triangleright$, and
    $\extraDW \triangleright \rightarrow \triangleright$. 
    Regarding the correspondence between $s_{i+1}$ and $s'_{i+1}$, the fact that this replacement is not applied to $s_i$ means that the extra symbol $\extraUP$ or $\extraDW$ has been removed. 
    
    \item Replacements with one extra symbol of the form
    $\extraDW \UP \rightarrow \UP \extraDW$, 
    $\extraUP\DW \rightarrow \DW \extraUP$,
    $\DW\extraUP \rightarrow \extraUP \DW$, and
    $\UP\extraDW\rightarrow\extraDW\UP$.
    Regarding the correspondence between $s_{i+1}$ and $s'_{i+1}$, the extra symbol $\extraDW$ or $\extraUP$ has just changed its position in the string $s_{i+1}$. 
    
    \item Replacements containing only extra-symbols are simply not taken into account.
\end{itemize}

Hence, the theorem follows by induction.
\end{proof}

\begin{corollary}\label{temp1}
For every string $s$, there exists a shortest rewriting sequence from $s$ that uses no replacement with the right-hand side longer than the left-hand side.
\end{corollary}
\begin{proof}
Consider a rewriting sequence $s'_0,\ldots,s'_{b(s)}$ that applies a replacement with the right-hand side longer than the left-hand side from $s'_i$ to $s'_{i+1}$ and let $i$ be the smallest such value. Let $s$ be the string obtained from $s'_i$ applying the subset of replacements that exclude the ones with the right-hand side longer than the left-hand side. We have that $s$ is a subword of $s'_{i+1}$. By Theorem~\ref{temp0}, $b(s) \leq b(s'_{i+1})$, and the corollary follows.
\end{proof}

\begin{corollary}\label{temp2}
For every string $s$ of length at least $3$, there exists a shortest rewriting sequence from $s$ where every rewriting contains two extreme replacements.
\end{corollary}
\begin{proof}
Consider a shortest rewriting sequence $s \Rightarrow s_1 \Rightarrow \cdots $ that does not start with a left removal $\triangleleft \UP \rightarrow \triangleleft$ or $\triangleleft \DW \rightarrow \triangleleft$.
Assume without loss of generality that the string $s$ starts with $\triangleleft \UP$, since the case of $\triangleleft \DW$ is analogous. The case of the extreme replacement involving $\triangleright$ is also analogous.
Consider the different rewritings from $s$ to $s_1$:

\begin{itemize}
\item If the rewriting is $\triangleleft \UP X \Rightarrow \triangleleft \UP Y$ for strings $X,Y$ where the leftmost $\UP$ is not involved in any replacement, then the length of the full rewriting sequence is $1 + b(\triangleleft \UP Y)$.
We now show how to obtain a rewriting of the same length that starts with an extreme replacement.
We rewrite $\triangleleft \UP X \Rightarrow \triangleleft Y$ providing a sequence of total length $1+b(\triangleleft Y)$. By Theorem~\ref{temp0}, we have $b(\triangleleft Y) \leq b(\triangleleft \UP Y)$.

\item If the rewriting is $\triangleleft \UP \UP X \Rightarrow \triangleleft \UP Y$ for strings $X,Y$ where the leftmost replacement is $\UP\UP \rightarrow \UP$, then we obtain the same string $\triangleleft \UP Y$ if we apply the replacement $\triangleleft\UP \rightarrow \UP$.

\item If the rewriting is $\triangleleft \UP \DW X \Rightarrow \triangleleft \DW \UP Y$ for strings $X,Y$ where the leftmost replacement is $\UP\DW \rightarrow \DW\UP$, then the length of the full rewriting sequence is $1 + b(\triangleleft \DW \UP Y)$.
We now show how to obtain a rewriting of the same length that starts with an extreme replacement.
We rewrite $\triangleleft \UP \DW X \Rightarrow \triangleleft \DW Y$ providing a sequence of total length $1+b(\triangleleft \DW Y)$. By Theorem~\ref{temp0}, we have $b(\triangleleft \UP Y) \leq b(\triangleleft \DW \UP Y)$.\qedhere
\end{itemize}
\end{proof}

We use a heuristic together with a depth first search exact solution to estimate $b(s)$. 
The heuristic is very involved and sometimes gives incorrect bounds. To fix some wrong bounds, we use an instance-independent preprocessing to compute tight values of $b(s)$ for small enough strings $s$ and store only the values where the heuristic incorrectly computes $b(s)$ in a file. This file is loaded by our solver and stored in a hash table.
}

\journal{
%....................................................................
\subsection{Improving a Solution} \label{s:improve}

Given a solution $P_1,\ldots,P_{|\TT|}$, we may improve it as follows. We choose a random path $P_i$ and use the SAT formulation to find a solution where the length of $P_i$ is decremented and all other lengths remain the same. We repeat this as often as necessary. Notice that the method may converge to a locally optimal minimal list of lengths that is not globally optimal.

If the number of clauses is too large, there are two different approaches that we can take (and they may be combined).
We may force the new solution to be close to the original one by only creating edge variables that cross few edges in the corresponding triangulation of the previous solution.

We may also trim the solution to a certain radius $r$, by only rebuilding the portion of the solution that is within $r$ steps from the center. In this case, it is helpful to use MaxSAT first, in order to reduce the number of unit-flips performed in the last steps as follows. Given a path $T_0,\ldots,T\ell$, we first use a MaxSAT solver to find the path of length $\ell$ that minimizes the number of unit flips performed from $T_{\ell-1}$ to $T_\ell$. The MaxSAT formulation is equal to the SAT formulation with soft clauses $\neg f(u,v,u_2,v_2,\ell)$ for each last step flip variable. We then find the path of length $\ell-1$ from $T_0$ to the $T_{\ell-1}$ of the previosu path that minimizes the number of unit flips performed from $T_{\ell-2}$ to $T_{\ell-1}$. We continue this way for $r$ steps.

%....................................................................
\subsection{Path Heuristic} \label{s:pathHeuristic}

The SAT formulation finds the shortest path connecting two triangulations reasonably fast, but it may be too slow for some usages. We also designed a heuristic that produces reasonably short paths quickly. We are given two triangulations $T_0,T'$ and the goal is to find a short path from $T_0$ to $T'$.

We use a greedy approach that iteratively obtains a triangulation $T_{i+1}$ from a triangulation $T_i$ as follows. Let $F$ denote the set of possible unit flips in $T_i$. More precisely, the elements of $F$ are pairs $e,e'$ of edges such that a unit flip from $T_i$ removes $e$ and inserts $e'$. We then build a graph $G(F)$ with vertex set $F$ and edges between two unit flips that share a triangle. Notice that the possible parallel flips correspond to independent sets in $G(F)$. We assign a weight to the vertices as follows. The weight of a vertex $e,e'$ is the number of edges in $T'$ crossed by $e'$ minus the number of edges in $T'$ crossed by $e$. Vertices of zero or negative weight are eliminated.

We then proceed to greedily find an independent set $I$ in $G(F)$. We iteratively add to $I$ the \emph{unmarked} vertex of maximum weight, breaking ties by minimum degree. We then \emph{mark} all the vertices in the closed neighborhood of $I$. We repeat until all vertices are marked.

This iterative approach is repeated until we reach $T'$, which will happen in $O(n^2)$ flips because of the following Theorem from~\cite{flipcross}.

\begin{theorem} \label{t:crossdec}
Let $T,T'$ be two triangulations of the same point set. If $T \neq T'$, then there exists a unit flip that replaces an edge $e \in T$ by an edge $e'$ such that $e$ crosses more edges of $T'$ than $e'$ does.
\end{theorem}

We further improve the heuristic using the squeaky wheel paradigm~\cite{squeaky}. Initially, we assign weight $1$ to every edge. We modify the definition of the weight of a flip as follows. The weight of a flip $e,e'$ is the sum of the weights of the edges in $T'$ crossed by $e'$ minus the sum of the weights of the edges in $T'$ crossed by $e$. When we reach triangulation $T_{i+1} = T'$ from a triangulation $T_i$ we increment the weight of all edges in $T'$ that are not in $T_{i+1}$. We repeat the greedy algorithm with the new weights, keeping the best solution found. This procedure is repeated multiple times.
}

\journal{
%....................................................................
\subsection{Solution Heuristic} ~\label{s:solutionHeuristic}

We use the following strategy to obtain reasonably good initial solutions that we may use as a basis to improve with the aforementioned methods. We start by adding the Delaunay triangulation as a candidate center. We then build other candidate centers by performing flips starting from the Delaunay triangulation.
Given an edge $e$ and a triangulation $T$, let $\chi(e,T)$ denote the number of segments of $T$ crossed by $e$.
We repeatedly perform unit flips that remove an edge $e$ and add an edge $e'$ maximizing
\[\sum_{T \in \TT} \chi(e)^p - \chi(e')^p,\]
as long as the value of the sum is positive. We add the triangulation we obtained to the set of candidate center and repeat the process for a different value of $p$.

We calculate the paths from the candidate centers to each solution, building a solution pool. For the next step it will be useful to have a small number of unit-flips in the last flip of the path from the input triangulation to the center. To do that, we either use the greedy heuristic or a MaxSAT formulation.

To improve the solution pool, we pick a solution from the pool and look at the triangulations that are one flip away from the center in each path. For each such triangulation, we compute the distance to the input triangulations and add the solution to the pool as before.
}

%--------------------------------------------------------------------
\section{Results} \label{s:results}
%--------------------------------------------------------------------

In this section, we present the computational results that we obtained with our implementation of the aforementioned algorithm. In Section~\ref{s:pathresults}, we present the results on computing short paths between two given triangulations. In Section~\ref{s:exactresults}, we present our exact solver, with and without the happy edges conjecture. 
\journal{In Section~\ref{s:heuristicresults}, we present the heuristics used to find solutions to the instances that we could not solve exactly.}

The solvers were coded in C++ and compiled with GCC and run a single thread. During the competition, they were executed on several Linux computers, either using GNU Parallel~\cite{parallel} for local executions or Slurm~\cite{slurm} for cluster executions. It was very useful to have access to machines with 128GB or more RAM to  solve large SAT formulations with \texttt{CaDiCal}~\cite{cadical}, which has been able to solve SAT instances with more than $50$ million variables and $500$ million clauses. The time measurements on this paper use an AMD Ryzen 9 9900X CPU and ASUS TUF B650M motherboard with 128GB of RAM running Fedora Core 43.

%....................................................................
\subsection{Path Calculation} \label{s:pathresults}

\conf{
\begin{table}[tb]
\centering
\begin{tabular}{c|cc|cc|cc}
 & \multicolumn{2}{c|}{Heuristic} & \multicolumn{2}{c|}{SAT happy} & \multicolumn{2}{c}{SAT exact}\\
$n$ & length & total time & length & time & length & time\\
\hline
500 & \textbf{12} & 70.9 & \textbf{12} & 1512 & \textbf{12} & 10191 \\
1000 & 12 & 77.6 & \textbf{11} & 6806 & \textbf{11} & 78264 \\
1500 & 12 & 100 & \textbf{11} & 19302 & \textbf{11} & 201353 \\
2000 & \textbf{13} & 208 &  \textbf{13} & 16937 & \textbf{13} & 485177 \\
3000 & \textbf{14} & 695 & \textbf{14} & 140729 & $\cdot$ & $\cdot$ \\
4000 & \textbf{15} & 1340 & \textbf{15} & 101402 & $\cdot$ & $\cdot$ \\
5000 & \textbf{16} & 656 & \textbf{16} & 1037840 & $\cdot$ & $\cdot$  \\
6000 & \textbf{16} & 1891 &  \textbf{16} & 271081 & $\cdot$ & $\cdot$  \\
7000 & 17 & 777 & \textbf{16} &  471389 & $\cdot$ & $\cdot$  \\
\end{tabular}
\caption{Length and computation time (in milliseconds) for a path from the Delaunay triangulation to the first input triangulation of the \texttt{rirs-}$n$-\texttt{-20} instance for different values of the number of points $n$. The Heuristic column shows the smallest length obtained by applying four heuristics (greedy (forward/backward) and squeaky-wheel (forward/backward)) described in the full version of the paper and the total time to run the four heuristics. The SAT columns correspond to SAT formulation with or without the happy edges conjecture, unless it takes too long.}
\label{tab:distance}
\end{table}
}

\journal{
\begin{table}[tb]
\centering
\begin{tabular}{c|cc|cc|cc|cc|cc|cc}
 & \multicolumn{2}{c|}{Greedy} & \multicolumn{2}{c|}{Greedy} & \multicolumn{2}{c|}{Squeaky} & \multicolumn{2}{c|}{Squeaky} & \multicolumn{2}{c|}{SAT} & \multicolumn{2}{c}{SAT}\\
 & \multicolumn{2}{c|}{forward} & \multicolumn{2}{c|}{backward} & \multicolumn{2}{c|}{forward} & \multicolumn{2}{c|}{backward} & \multicolumn{2}{c|}{happy} & \multicolumn{2}{c}{exact}\\
$n$ & $\ell$ & t & $\ell$ & t & $\ell$ & t & $\ell$ & t & $\ell$ & t & $\ell$ & t\\
\hline
500 & 13 & 4.2 & \textbf{12} & 3.7 & 13 & 45 & \textbf{12} & 18 & \textbf{12} & 1512 & \textbf{12} & 10191 \\
1000 & 12 & 7.8 & 12 & 6.8 & 12 & 25 & 12 & 38 & \textbf{11} & 6806 & \textbf{11} & 78264 \\
1500 & 13 & 13 & 12 & 12 & 12 & 52 & 12 & 23 & \textbf{11} & 19302 & \textbf{11} & 201353 \\
2000 & 15 & 21 & \textbf{13} & 18 & \textbf{13} & 65 & \textbf{13} & 104 & \textbf{13} & 16937 & \textbf{13} & 485177 \\
3000 & 15 & 36 & 15 & 33 & \textbf{14} & 193 & 15 & 433 & \textbf{14} & 140729 & $\cdot$ & $\cdot$ \\
4000 & 18 & 57 & 16 & 51 & 16 & 577 & \textbf{15} & 655 & \textbf{15} & 101402 & $\cdot$ & $\cdot$ \\
5000 & 18 & 70 & 17 & 70 & \textbf{16} & 238 & 17 & 278 & \textbf{16} & 1037840 & $\cdot$ & $\cdot$  \\
6000 & 17 & 93 & \textbf{16} & 83 & 17 & 1071 & \textbf{16} & 644 & \textbf{16} & 271081 & $\cdot$ & $\cdot$ \\
7000 & 17 & 94 & 17 & 96 & 17 & 180 & 17 & 407 & \textbf{16} & 471389 & $\cdot$ & $\cdot$  \\
\end{tabular}
\caption{Length and computation time (in milliseconds) from the Delaunay triangulation to the first input triangulation of the \texttt{rirs-}$n$-\texttt{-20} instance for different values of the number of points $n$. The columns respectively correspond to the greedy heuristic forward and backward, the squeaky wheel heuristic forward and backward, the SAT solution with the happy edges conjecture, and the SAT solution without the happy edges conjecture, unless it takes too long. \journal{The best lengths found are shown in bold.}}
\label{tab:distance}
\end{table}
}

Computing short paths between two given triangulations is a key component to obtain good solutions. Typically, these paths are computed with an input triangulation as one extreme, and a triangulation that makes a reasonably good center as the other extreme. 
\journal{In this section, we use the Delaunay triangulation as one extreme, because it is a well defined triangulation that makes a reasonably good (but not very good) center. Table~\ref{tab:distance} shows the length of the path and the running time of different heuristics and SAT solutions with and without the happy edges conjecture. The paths are calculated from the Delaunay triangulation to the first input triangulation of several instances. The SAT paths are obtained by first running the squeaky wheel heuristic in both directions, and then iteratively decreasing the path length. Notice that the running time of the heuristics is much smaller, while the result is rarely more than $1$ unit away from the optimal distance, especially if we take the minimum of both directions. We run the squeaky wheel heuristic for at most $16$ iterations, but stop earlier if the length increases from one iteration to the next.}
\conf{Table~\ref{tab:distance} shows the best length of the path obtained with our heuristics (described in the full version of the paper) with the associated running times compared to our SAT formulation with and without the happy edges conjecture. The paths are computed from the Delaunay triangulation that makes a reasonably good (but not very good) center to the first input triangulation of several instances. 
The SAT paths are obtained by first running a heuristic in both directions, and then iteratively decreasing the path length with our SAT solution.}

%....................................................................
\subsection{Exact Solutions} \label{s:exactresults}

\begin{figure}[ht]
 \centering
 \includegraphics[height=5.3cm]{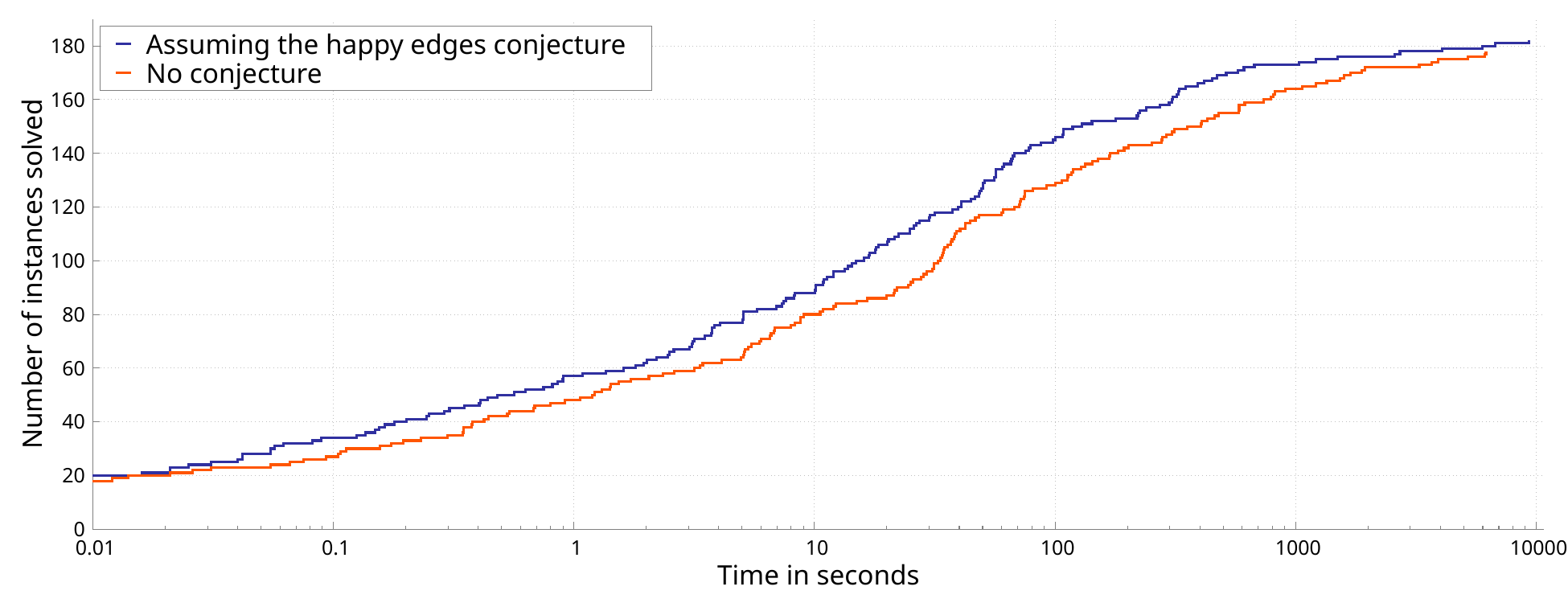}
 \caption{Number of exact solutions found as a function of the running time over 3 hours of execution with and without assuming the happy edges conjecture.}
 \label{f:exact}
\end{figure}

\conf{We end the paper with the performance of our exact solver.}
\journal{We now present the exact solver that we used to solve most instances exactly.}
Figure~\ref{f:exact} plots the number of exact solutions found as a function of runtime, with and without the happy edges conjecture. Since the solution values are identical in both cases, the happy edge conjecture holds.
\journal{Notice that the impact of the happy edges conjecture is less significant than when computing only paths, as there are few common edges on all input triangulations. We remark that during the competition, we managed to find exact solutions to $189$ of the $201$ \texttt{small} instances without assuming the happy edges conjecture.}

\journal{
%....................................................................
\subsection{Crossing Lower Bound} \label{s:corsslbresults}

The heuristic that estimates $b(s)$ for a string $s$ (see Section~\ref{s:crossLB}) is surprisingly accurate.  Figure~\ref{f:crosslbplot} shows the heuristic error rates for different string lengths. The smallest strings where the heuristic fails have length $12$ and are:\\
$\triangleleft ? \UP \DW \DW \UP \DW \DW \DW \DW ? \triangleright$, 
$\triangleleft ? \DW \UP \UP \DW \UP \UP \UP \UP ? \triangleright$,
$\triangleleft ? \DW \DW \DW \DW \UP \DW \DW \UP ? \triangleright$, and
$\triangleleft ? \UP \UP \UP \UP \DW \UP \UP \DW ? \triangleright$,\\
with the question marks replaced by $\UP$ or $\DW$, which gives an error rate of $4 \cdot 4 / 2^{10} \approx 0.016$ as the two extreme symbols are fixed.
We compare the error rate when the heuristic uses exact preprocessed values to solve subproblems and when the heuristic does not use any preprocessed data. The heuristic was never wrong by more than one unit on strings of up to $28$ symbols.

\begin{figure}[htb]
 \centering
 \includegraphics[height=5.3cm]{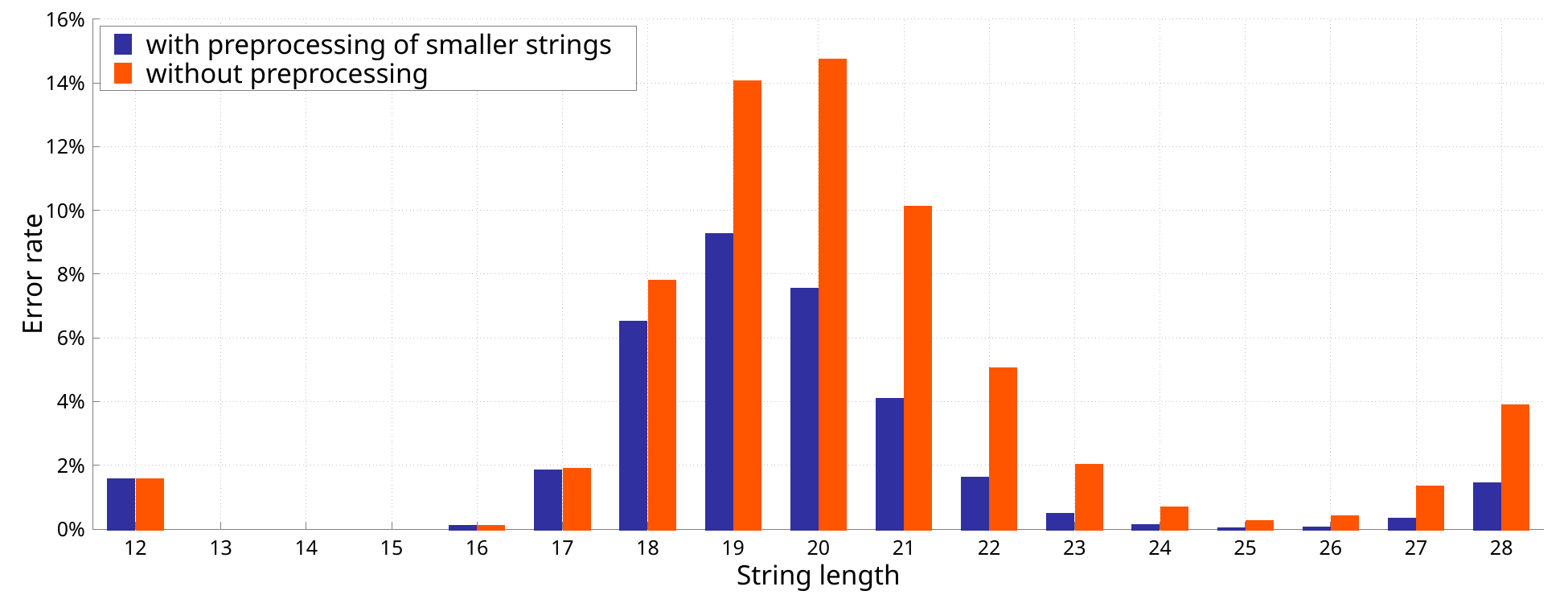}
 \caption{Error rate of the crossing lower bound for strings with $k$ symbols, which corresponds to $k-1$ crossings.}
 \label{f:crosslbplot}
\end{figure}

%....................................................................
\subsection{Heuristic Solutions} \label{s:heuristicresults}

The typical process to solve instances that we did not solve exactly consists of several steps. First, we use the heuristic from Section~\ref{s:solutionHeuristic} to obtain a reasonably good center. Notice that no SAT solver is used in this part. The evolution of the objective value for this step is shown in Figure~\ref{f:solveplot}. The name of the \texttt{rirs} instances is composed of two values. The first is the number of points and the second is the number of input triangulations.

\begin{figure}[ht]
 \centering
 \hfill
 \includegraphics[height=5.3cm]{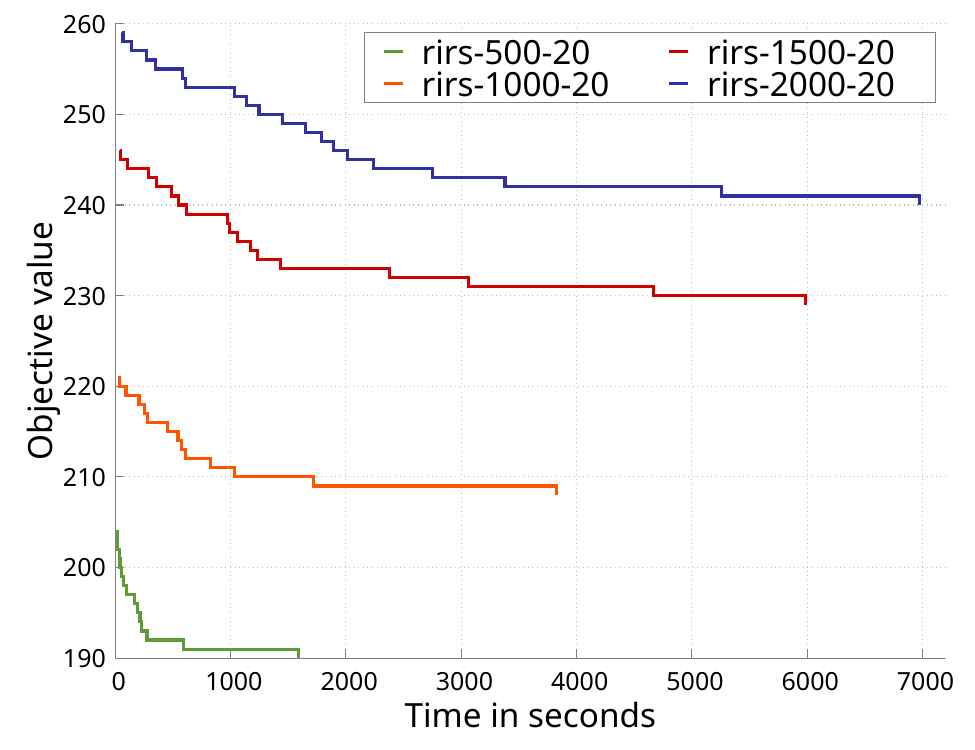}
 \hfill
 \includegraphics[height=5.3cm]{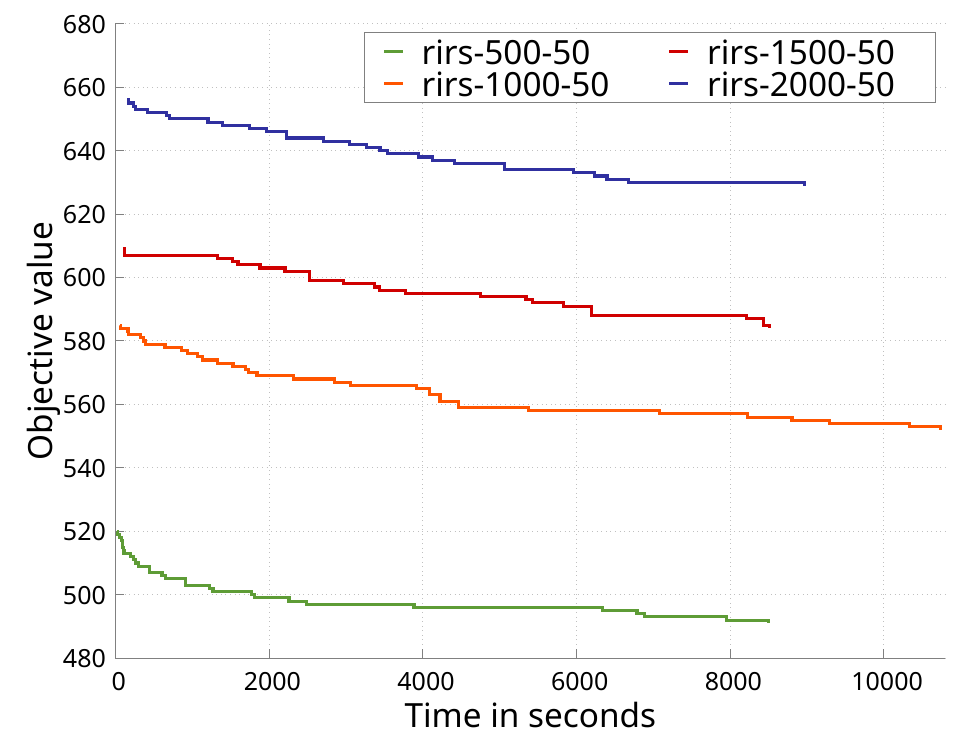}
 \hfill{}
 \caption{Evolution of the best solution found by the heuristic solver without any SAT solver for different instances.}
 \label{f:solveplot}
\end{figure}

Second, we build a new solution that keeps the same center but we recalculate the paths using the SAT formulation from \ref{s:satDistance}, assuming the happy edges conjecture (Section~\ref{s:happy}) and inexact lower bounds (Section ~\ref{s:crossLB}). This will typically reduce the objective value by $1$ to $4$ percent. The calculation of $50$ paths in each solution is shown in Figure~\ref{f:recalcplot}. Notice that the running time increases rapidly with the number of points and that finding a shorter path (satisfiable SAT problem) is typically slower than when no shorter path exists (unsatisfiable SAT problems).

\begin{figure}[ht]
 \centering
 \includegraphics[height=5.3cm]{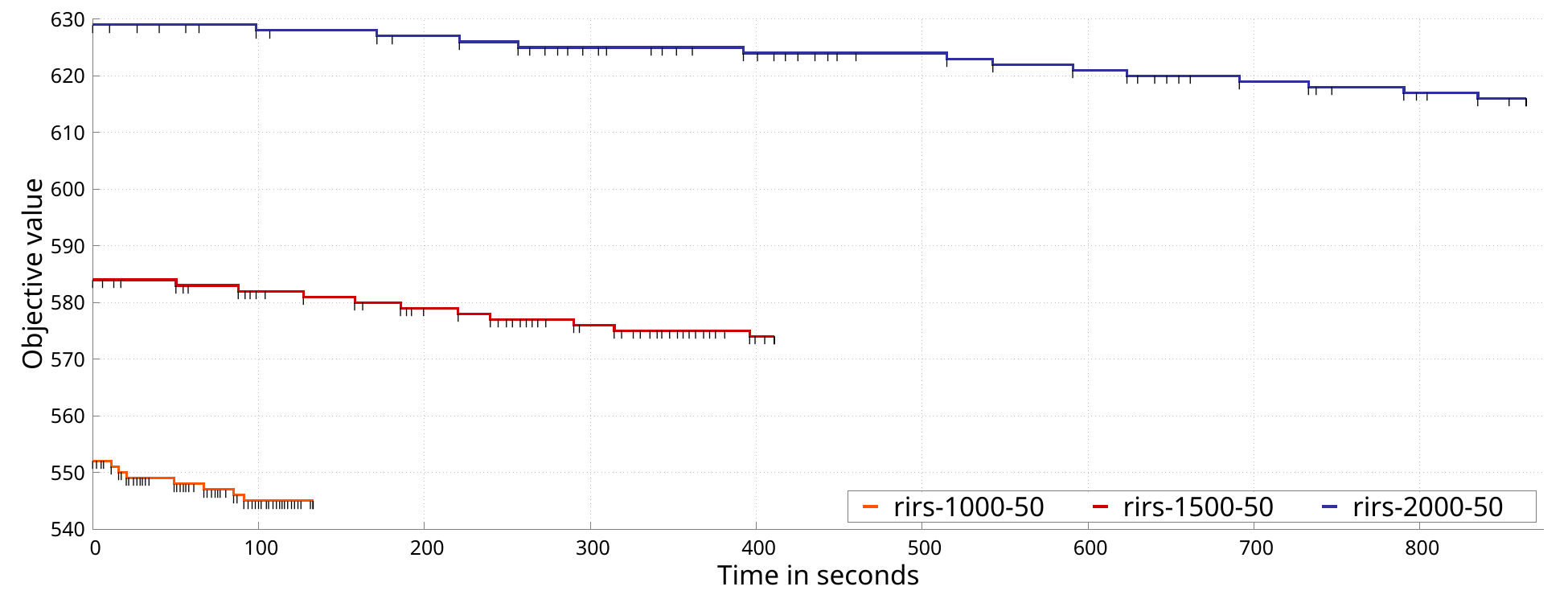}
 \caption{Path by path improvement of the heuristic solutions using a SAT solver, where each path is recalculated but the center is kept unchanged. Each black vertical bar represents the beginning of the computation of a path.}
 \label{f:recalcplot}
\end{figure}

Third, we improve the solution using the SAT formulation from Section~\ref{s:satSolution}, assuming the happy edges conjecture (Section~\ref{s:happy}) and inexact lower bounds (Section ~\ref{s:crossLB}). This part requires adjusting a large number of parameters in order to obtain SAT problems that are not too hard. The parameters include the distance to the previous center, the distance to the previous path, and whether the solution will be trimmed. Optionally but recommended when the problem is trimmed, we use MaxSAT to reduce the number of unit flips close to the center. The improvement of some solutions over time is shown in Figure~\ref{f:improveplot}. Notice that there is a significant preprocessing time to calculate edge variables with a limited number of crossings. This preprocessing is applied initially and after every improvement. Also notice that unsatisfiable SAT problems, which mean no improvement in the solution, are solved much faster than satisfiable ones. 

\begin{figure}[htb]
 \centering
 \includegraphics[height=5.3cm]{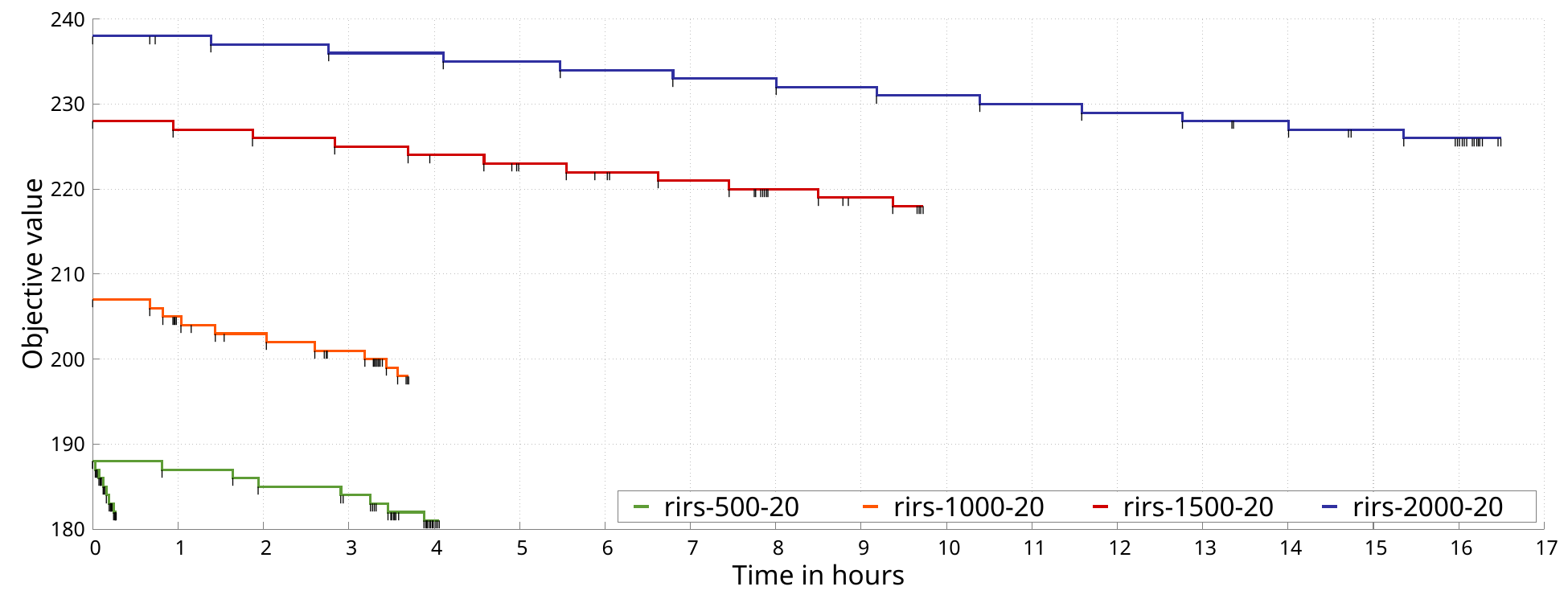}
 \includegraphics[height=5.3cm]{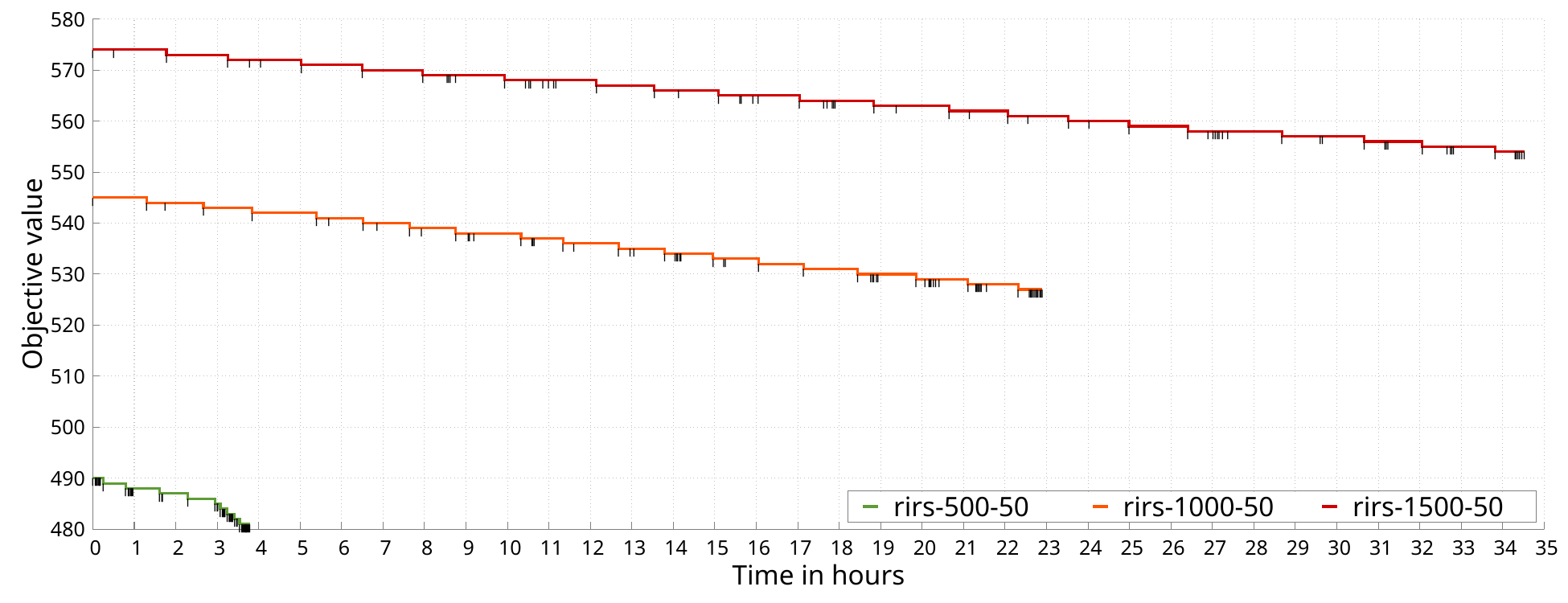}
 \caption{Improvement of the whole solution using a SAT solver and reducing the length of a random path by one unit at a time. We constrain the edges in the new solution to cross at most $3$ edges of the corresponding triangulation in the solution (except for the instance \texttt{rirs-500-20}, where the slower but more effective execution with the parameter set to at most $7$ edges is also pictured) but do not use trimming. Each black vertical bar represents a new path length that we try to decrement.}
 \label{f:improveplot}
\end{figure}
}

\journal{
%--------------------------------------------------------------------
\section{Conclusion and Open Problems} \label{s:conclusion}
%--------------------------------------------------------------------

We were surprised that we managed to solve so many instances exactly and how well an heuristic approximates the exact distance. We believe the following factors help explain the strong practical performance:

\begin{itemize}
    \item The short path length that allows a somewhat small number of variables.
    \item A SAT model where most clauses have size $2$.
    \item The ability to eliminate many variables through several arguments.
    \item The fact that the happy edges conjecture is either true or at least holds in most practical cases.
\end{itemize}

The short path lengths are somewhat surprising, given that an $\Omega(n)$ lower bound is presented in~\cite{pflips}, in contrast to the $\Theta(\log n)$ bound for combinatorial triangulations~\cite{pflipscomb}, where we are allowed to flip non-convex quadrilaterals. In the challenge instances, we observed path lengths that are roughly $O(\log n)$.

Still, there are \texttt{random} instances with only $160$ points and $20$ input triangulations and \texttt{woc} instances with only $185$ points and $6$ triangulations that we could not solve exactly. Also, we still managed to improve solutions to instances with as few as $320$ points and $20$ input triangulations months after the beginning of the challenge.

The challenge instances did not have many points in convex position. Surprisingly, the case where all points are in convex position is the hardest for our SAT formulation, as there are $\Theta(n^4)$ empty convex quadrilaterals. We wonder if a different model works better when all and also when most points are in convex position.

We believe that the same problem with unit flips is significantly harder because the long path lengths make the SAT formulation much more complex and removing a happy edge can reduce a path length from $\Theta(n^2)$ to $\Theta(n)$.

Many theoretical open problems remain, such as proving Conjecture~\ref{c:happy}. We also wonder if parallel flip distance problem is NP-hard in general and in convex position, as well as the unit flip distance problem in the convex case (the problem is NP-hard for general point sets~\cite{Pil14}). We also do not know if the problem of calculating $b(s)$ (see Section~\ref{s:crossLB}) can be solved in polynomial time, possibly using dynamic programming.
}

\bibliography{references}

\end{document}